\documentclass[letterpaper, 10 pt, conference]{ieeeconf}  

\IEEEoverridecommandlockouts                              
\usepackage{amsmath} 
\usepackage{amssymb}
\usepackage{amsfonts}
\usepackage{mathrsfs} 
\usepackage{xcolor}
\usepackage{algorithm}
\usepackage{algorithmic}
\newtheorem{assumption}{Assumption}
\newtheorem{theorem}{Theorem}
\newtheorem{lemma}{Lemma}
\newtheorem{proposition}{Proposition}
\newtheorem{problem}{Problem}

 \newtheorem{Definition}{Definition}

\title{\LARGE \bf
Safety-Aware Learning-Based Control of Systems with Uncertainty Dependent Constraints (Extended Version$^{\star}$)
}

\author{Jafar Abbaszadeh Chekan and C\'edric Langbort
\thanks{This paper is based upon work supported by the National Science Foundation under Grant No 2007604 as well as a C3.AI Digital Transformation Institute Grant.}
\thanks{J.~A.~Chekan and C.~Langbort (emails: jafar2 \& langbort@illinois.edu) are with the Coordinated
  Science Laboratory and the Department of Aerospace
  Engineering at the University of Illinois at Urbana-Champaign (UIUC).}%
}

\begin{document}

\maketitle
\thispagestyle{empty}
\pagestyle{empty}

\begin{abstract}
In this paper, we tackle the problem of safely stabilizing an originally (partially) unknown system while ensuring that it does not leave a prescribed 'safe set'  whose structure itself depends on the unknown part of the system's dynamics. For this aim, we apply a popular approach based on control Lyapunov functions (CLF), control barrier functions (CBF), and Gaussian processes (to build confidence set around the unknown term), which has proved successful in the known-safe set setting. However, with the mentioned safety set structure, we witness the introduction of higher-order terms to be estimated and bounded with high probability using only system state measurements. In this paper, we build on the recent literature on Gaussian Processes (GPs) and reproducing kernels to address the challenge and show how to modify the CLF-CBF-based approach correspondingly to obtain safety guarantees. To overcome the intractability of verification of these conditions on the continuous domain, we apply discretization of the state space and use Lipschitz continuity properties of dynamics to derive equivalent CLF and CBF certificates in discrete state space. Finally, we discuss the strategy for the control design aim using the derived certificates.
\end{abstract}

\section{INTRODUCTION}


Guaranteeing the stability of a closed-loop system as learning is taking place is a crucial challenge in applying learning theory-based controls. The presence of additional constraints on state and/or input signals further compounds matters. For example, in robotics applications, such constraints could represent the geometric and dynamic constraints (e.g., walls and obstacles in a robot's environment) whose violation can cause severe damage to the system and its environment. Though most existing works trying to address the safety-aware learning-based control problem assume known safety constraints, in this work, our goal is to address the safety concern when the safety set itself depends on dynamic uncertainty. More precisely, our goal is to design an algorithm that enables joint stabilization of the system and invariance of a safety set while learning the uncertainty that affects both dynamics and the set. Managing safety in this setting brings complexities and challenges that need rigorous analysis, which we discuss after providing a quick review of the literature.

The stability concern is addressed by using the Control Lyapunov Function (CLF) constraints through which the system is pushed to reside in the region of attraction that is estimated and updated in time. In contrast to a large body of works in the literature trying to learn the Lyapunov function when the system dynamics is known  (see for example \cite{giesl2016approximation, richards2018lyapunov,ravanbakhsh2019learning}), we focus on deriving the region of attraction and constraining control design to satisfy the obtained set when the system dynamics is uncertain and Lyapunov function is given. Finding a safe region of attraction in real-time when the system dynamics is uncertain has received tremendous attention. For example, \cite{akametalu2014reachability} introduce a reachability-based approach to learning the safe region. Furthermore, \cite{berkenkamp2017safe} proposes a model-based RL which guarantees safety through the learning region of attraction in discrete time. An analogous study also has been carried out in a continuous time nonlinear system by applying the Gaussian process by \cite{berkenkamp2016safe}. Our analysis generalizes the analysis of \cite{berkenkamp2016safe} for the Exponential Control Lyapunov Function (ECLF) case.

As for the imposed state-dependent constraints satisfaction, it is usually addressed using Control Barrier Function constraints in control design. Recently, the control barrier function, a Lyapunov-like function, has contributed to the safety aspect of control systems mostly in known dynamic setting \cite{ames2016control, xiao2019control, nguyen2016exponential}. Inclusion of CBF type of constraints in control design guarantees forward invariance of a given set, i.e., starting from a given set, the system states will reside in it forward in time. CBF has shown its effectiveness in various applications such as robotics and automotive systems. Authors in \cite{ames2014control} have applied a CBF constraint in a quadratic program to address the adaptive cruise control problem (ACC), whereas later on \cite{xiao2019control} by presenting failure of first order CBF constraints through an example introduced a novel high order CBF (HCBF) to tackle this problem. In analogous to HCBF, which is introduced for the first time by \cite{nguyen2016exponential}, CLF has been extended to higher order scheme in \cite{wang2021chance} and  \cite{bahreinian2021robust} along the way of path planning in the polygonal environment using CBF and CLF. 

Most recently CBF has been combined with learning techniques to design control for safety critical systems \cite{taylor2020learning, emam2021safe, wabersich2021predictive}. Authors in \cite{cheng2019end} applied Gaussian Process (GP) to model the uncertainty using which they imposed an uncertainty-aware CBF constraint to ensure safety of an RL algorithm. 

GPs, is a tractable regression method to estimate an unknown function given a set of
(noisy) measurements of its values at some  points. GPs as a widely used nonparametric estimation technique that is widely applied in sequential decision making problems such as multi-armed bandit problem (\cite{williams2006gaussian}, \cite{srinivas2009gaussian}, and \cite{chowdhury2017kernelized}). This technique together with reproducing kernels are used to quantify the uncertainty of an unknown function through a confidence interval constructed using the mean and variance of estimation. Recently GP has been successfully applied to ROA estimation by \cite{berkenkamp2016safe}, \cite{berkenkamp2017safe} and CBF constraints construction for systems with known safety constraint by \cite{cheng2019end}.

As an extension to the orthodox research body with known safety constraints assumption, our work is concerned with addressing the control design problem when the dynamic uncertainty explicitly appears in the safety constraints. This class of constraints is usual in controls (see for example \cite{vignali2018data}).

Having uncertainty in the imposed constraints adds to the existing complexities of guaranteeing forward in-variance of them as this addition requires estimating terms with higher order derivatives of uncertainty. To overcome this challenge we extend the existing results in GPs and reproducing kernels in bandit literature, and build high probability confidence sets for these additional terms using only system state measurement. Using the obtained sets which are updated in real time we derive second order relative degree ECBF constraint. For safe control design we use these constraints together with ECLF constraint. 

The last challenge is computation-wise, as checking the obtained ECLF and ECBF certificates (in the continuous domain) in the control design procedure is a computational burden. We overcome this challenge by exploiting  Lipschitz continuity of dynamics which helps to equivalently check the constraints on a finite grid on state space. 

Our proposed methodology can successfully address safety for the setting tackled by \cite{wang2021learning} in which the dynamics is uncertain, the safety set is known, but higher order CBF certificates are needed. For the mentioned setting even though the model (uncertain dynamics but known safety constraints) is benign, deriving higher order CBF certificates requires the estimate of uncertainty derivatives, which is the challenge resolved by our study. To build the confidence sets of higher order derivatives we build on the techniques of \cite{chowdhury2017kernelized}.

We delineate the remainder of our paper as follows. First, we briefly restate the prominent basics of the literature for higher order ECBF and ES-CLF for known systems in Section \ref{subsecKnown}. Then, in Section \ref{Sec:problem statement}, we derive these certificates for the partially unknown system with uncertain safety set in the continuous domain given the confidence set of unknown term and its derivatives. Next, section \ref{sec:Gaussian} provides an overview of GPs and the construction of associated confidence sets. Then, section \ref{sec:DiscreteECLFandCBF}, to overcome computational burden, derives equivalent certificates in discrete domain using which we discuss a strategy for control design. Finally, we conclude the key achievements in Section \ref{Sec:Conclus}. The technical and complimentary proofs are provided in Appendix of \cite{chekan2022Safetyaware}.

\section{PRELIMINARIES} 
\label{Sec:preliminaries}

Consider a nonlinear input-affine control system of the form

\begin{equation} 
\dot{x}= \underbrace{f(x)+g(x)u}_\text{known model}+\underbrace{d(x)}_\text{unknown term} \label{eq:model}
\end{equation}
where $x\in \mathcal{X} \subset \mathbb{R}^n$, $f:\mathbb{R}^n\rightarrow \mathbb{R}^n$, $u\in \mathcal{U} \subset \mathbb{R}^m$, and $g:\mathbb{R}^n\rightarrow \mathbb{R}^{n\times m}$. The dynamics' model $f(x)$ and $g(x)\in \mathbb{R}^n$ are known a priori while $d(x)$ is unknown. Given a set
\begin{equation}\label{eq:set}
\mathcal{K}=\{x\in \mathbb{R}^n|\; h(x)+\mathcal{F}(d(x))\geq 0\},
\end{equation}
for some known $h:\mathbb{R}^n\rightarrow \mathbb{R}$ and known operator $\mathcal{F}:\mathbb{R}^n\rightarrow \mathbb{R}$, our high-level goal is to design a control signal $u$ that stabilizes the system while keeping the state in $\mathcal{K}$ in a sense to be made clear in the next sections. We emphasize that, unlike in existing works, the safety constraint defining set (\ref{eq:set}) is dependent on the a priori unknown part of model $d(x)$.

In order to define our goal more rigorously, we start by reviewing some tools for ensuring stability and (forward) invariance of safety set in the case of fully known dynamics and constraint sets by following the lead works \cite{nguyen2016exponential}, \cite{bahreinian2021robust} and \cite{wang2021chance}. We then detail our strategy to employ these tools using approximate representations of non-directly accessible quantities when part of the system dynamics is unknown and the safety set is uncertain. Our approach extends the lead of works such as \cite{berkenkamp2016safe} and \cite{berkenkamp2017safe} which approximate control Lyapunov function, and \cite{cheng2019end} that approximates a first order control barrier function. 

\subsection{ES-CLF and ECBF for Known Systems}\label{subsecKnown}

Let us assume, for the purposes of this subsection only, that the right hand side of (\ref{eq:model}) is fully known and that we are interested in stabilizing the system while rendering the set  $\mathcal{K}$ given by (\ref{eq:set}) forward invariant in closed-loop. Two notions, borrowed from \cite{bahreinian2021robust} and explicited in definitions \ref{def1}- \ref{def2}, are useful to this end. We refer the reader to this work and references therein for reminders regarding the classical notions of Lie derivatives, relative degree and transverse dynamics employed in the definitions.

\begin{Definition}{(ES-CLF):} A differentiable function $V: \mathcal{X} \rightarrow \mathbb{R}$ such that $V(x)\geq 0$ for all $x \in \mathcal{X}$ is called an Exponentially Stabilizing Control Lyapunov Function (ES-CLF) if there exists  $K_V>0$ such that, for all $x \in \mathcal{X}$, 
\begin{equation}
\label{eq:ECLF}
    \mathcal{L}_{f+d}V(x)+\mathcal{L}_gV(x)u+K_VV(x)\leq 0
\end{equation}
for some $u \in \mathcal{U}$.
\label{def1}
\end{Definition}

\begin{Definition} {(ECBF):}
 Let a function $H: \mathcal{X} \rightarrow \mathbb{R}$ have relative degree  $r_H$ with respect to the system (\ref{eq:model}) and let the transversal dynamics  be written as follows:
\begin{equation}
\begin{array}{ll}
\dot{\eta_H}=F \eta_H+G\mu \\
 H(x)=C\eta_H
\end{array}
\end{equation} 
where the transversal state $\eta_H$ is defined as
 \begin{equation}
\eta_H= \begin{pmatrix}
H(x)  \\
\mathcal{L}_{f+d}H(x)\\
\mathcal{L}^2_{f+d}H(x)\\
. \\
.\\
\mathcal{L}^{r_H-1}_{f+d}H(x)
\end{pmatrix}. \label{eq:Eta_h}
\end{equation}
  We say that $H$ is an Exponential Control Barrier Function (ECBF) if there exists a  $K_H\in \mathbb{R}^{r_H}$ such that, for all $x \in \mathbb{R}^n$ with $H(x) >0$, 
\begin{equation}
    \mathcal{L}^{r_H}_{f+d}H(x)+\mathcal{L}_g\mathcal{L}^{r_H-1}_{f+d}H(x)u+K_H^\top\eta_H\geq 0\ \label{eq:ECBF}
\end{equation}
for some control $u\in \mathcal{U}$.
\label{def2}
\end{Definition}

As explained in \cite{bahreinian2021robust}, the reason the notions of ES-CLF and ECBF are useful is that they provide a constructive way to obtain a stabilizing control law with desirable invariance properties as stated in the following propositions:
\\
\begin{proposition}
\label{prop1}
Let system (\ref{eq:model}) admit an ECBF $H$ such that the control law $\mu_H=-K_H\eta_H$ stabilizes the transverse dynamics, and let function $\pi: \mathbb{R}^n \rightarrow \mathcal{U}$ be such that $u = \pi(x)$ satisfies (\ref{eq:ECBF}) for all $x$. 
Then the control law $\pi$ renders the set $\mathcal{E} = \{x \in \mathbb{R}^n | H(x) \geq 0\}$ positively invariant in closed-loop, i.e., $x(t) \in \mathcal{E}$ for all $t \geq 0$, provided $x(0) \in \mathcal{E}$.
\end{proposition}

\begin{proposition}
\label{prop2}
Let system (\ref{eq:model}) admit an ES-CLF $V$ with $V(0)=0$ and let function $\pi: \mathbb{R}^n \rightarrow \mathcal{U}$ be such that $u = \pi(x)$ satisfies (\ref{eq:ECLF}) for all $x \in \mathcal{C}(c) = \{ x \in \mathbb{R}^n \; | \; V(x) \leq c \}$. Then the control law $\pi$ renders the origin exponentially stable in closed-loop within the level set $\mathcal{C}(c)$.
\end{proposition}

In order to keep the complexity of the derivations to a minimum, while emphasizing the main steps of our approach, we henceforth restrict ourselves to ECBF of relative degree $r_H=1$ as well as to a state dimension of $n=1$. Since we are interested in making the set $\mathcal{K}$ defined in (\ref{eq:set}) invariant, we also specialize the results above to the case where the ECBF of interest is $H(x) = h(x)+\mathcal{F}(d(x))$. Under these additional assumptions  the conditions to be checked by a control law $u = \pi(x)$ to achieve the guarantees of propositions \ref{prop1} and \ref{prop2} are
\begin{equation}
\begin{array}{ll}
 &\dot{\mathscr{H}}(x,\pi(x)):=(\frac{\partial h(x)}{\partial{x}}+\frac{\partial \mathcal{F}(d(x))}{\partial{d(x)}} \frac{\partial d(x)}{\partial{x}})(f(x)+d(x))\\
&+(\frac{\partial h(x)}{\partial x}+\frac{\partial \mathcal{F}(d(x))}{\partial{d(x)}} \frac{\partial d(x)}{\partial{x}}) g(x) \pi(x)+\\
&K_H(h(x)+\mathcal{F}(d(x)))\geq 0
\end{array} 
\label{SafetyCond}
\end{equation}
and
\begin{align}
  \nonumber \mathcal{\dot{V}}(x,\pi(x)):=&\frac{\partial V(x)}{\partial x}^\top (f(x)+g(x)\pi(x)\\
  &+d(x))+K_VV(x) \leq 0, \label{eq:stabilityy}
\end{align}
respectively. A further simplification we will make from here on, for the sake of brevity in derivation, is focus on the case $\mathcal{F}(d(x))=M d(x)$ with a known $M\in \mathbb{R}$. Similar results can be derived in the general case, at the cost of more intricate formulae.

\subsection{ES-CLF and ECBF for Partially Unknown Systems with Uncertain Safety Set
} \label{Sec:problem statement}
 
 When system (\ref{eq:model}) is partially unknown, conditions (\ref{SafetyCond}) and (\ref{eq:stabilityy}) cannot be used as-is, because they involve the unknown term $d(x)$, its derivative $\partial d(x)/\partial x$ and their product. To estimate these quantities, we apply Gaussian process, a non-parametric approach, to build a high probability confidence interval around them relying on state measurements. While the technical details of this approach are provided in the next section, the most important aspect of this approximation process for our purposes is that, after 
$n-1$ measurements, all quantities of interest are guaranteed to lie in so-called confidence sets  with high probability, i.e. 
\begin{equation}
\label{set_GP}
\mathcal{G}(x) \in I^{n-1}_{\mathcal{G}}(x) := \left\{y \; | \;|y-\mu_{\mathcal{G}}^{n-1}(x)|\leq\sqrt{\beta^n_{\mathcal{G}}} \sigma^{n-1}_{\mathcal{G}}(x)\right\}
\end{equation}
with probability at least $1-\delta$.

In (\ref{set_GP}), $\mathcal{G}(x)$ designates $d(x)$, $\partial d(x)/\partial x$ or $d(x)\partial d(x)/\partial x$, while $\mu_{\mathcal{G}}^{n-1}(x)$ and $\sigma^{n-1}_{\mathcal{G}}(x)$  are quantities that can be computed directly from the $n-1$ measurements. The factor $\beta^n_{\mathcal{G}}$'s is a function of $\delta$ to be derived in Theorems 1-3.

%

%

Because we only have access to an estimate of the term $d(x)$, we similarly cannot use the level set $\mathcal{E}$ of Proposition \ref{prop1} as-is for control design. Instead, we introduce 
\begin{align}
    \mathcal{E}_n= \left \{x\in \mathbb{R}^n\; |\; h(x)+ \min_{y \in I^n_{d}(x)} \mathcal{F}(y) \geq 0 \right\} \label{eq:safeSetnMeasu}
\end{align}

which is constructed by using the confidence set of $d(x)$ estimated with $n$ samples. This level set is an inner approximation of the true level set $\mathcal{E}$, (i.e., $\mathcal{E}_n\subseteq \mathcal{E}$), whose size increases with $n$ (i.e., $\mathcal{E}_{n-1}\subseteq \mathcal{E}_n$).  In turn, as more and more data is collected, an increasingly tighter approximation of the desired invariant set is constructed.



With these facts and notations in place, we can now rigorously state our control design goal, namely:
\begin{problem}
\label{prob1}
Given an ES-CLF $V$ and the GP estimates of $d(x)$, find a control policy $\pi$ which, for every $n$,  maximizes the size of the region $(\mathcal{E}_n\cap \mathcal{C}(c)) \cap \mathcal{X}$.
\end{problem}

%
 %

A sufficient condition to render this latter set invariant, provided $V$ is a ES-CLF, is to ensure that a upper-bound to (\ref{eq:stabilityy}) is non-positive for all $x$ belonging to it. Using GP estimates of $d(x)$, $\partial d(x)/\partial x$, and $d(x)\partial d(x)/\partial x$, we can show that $\dot{\mathcal{H}}(x, \pi(x))\in [l_{\dot{\mathcal{H}}}^n(x, \pi(x)),\;u_{\dot{\mathcal{H}}}^n(x, \pi(x))]$ where
\begin{align}
      \nonumber & l_{\dot{\mathscr{H}}}^n(x, \pi(x))=\frac{\partial h}{\partial x}\big(f(x)+\mu_d^n(x)+g(x)\pi(x)\big)-\\
    \nonumber &\bigg|\frac{\partial h}{\partial x}\bigg|\sqrt{\beta_d^n}\sigma_d^n(x) +M\bigg(f(x)+g(x)\bigg)\mu_{\partial d}^n-\\
     \nonumber &\bigg|M(f(x)+g(x) \pi(x))\bigg|\sqrt{\beta_{\partial d}^n}\sigma_{\partial d}^n+M\mu_{d\partial d}-\\
   \nonumber& |M|\sqrt{\beta_{d\partial d}^n}\sigma_{d\partial d}^n+K_H h(x)+\min_{y\in I^n_d} K_H M y
   \end{align}
and
\begin{align}
   \nonumber & u_{\dot{\mathscr{H}}}^n(x, \pi(x))=\frac{\partial h}{\partial x}\big(f(x)+\mu_d^n(x)+g(x)\pi(x)\big)+\\
    \nonumber & \bigg|\frac{\partial h}{\partial x}\bigg|\sqrt{\beta_d^n}\sigma_d^n(x)+M\bigg(f(x)+g(x)\bigg)\mu_{\partial d}^n+\\
   \nonumber &\bigg|M(f(x)+g(x) \pi(x))\bigg|\sqrt{\beta_{\partial d}^n}\sigma_{\partial d}^n+M\mu_{d\partial d}+\\
     \nonumber&|M|\sqrt{\beta_{d\partial d}^n}\sigma_{d\partial d}^n+K_H h(x)+\max_{y\in I^n_d} K_HM y
\end{align}


It is thus enough, for the purposes of rendering $\mathcal{E}_n$ invariant, to ensure that  $l_{\dot{\mathscr{H}}}^n(x, \pi(x))\geq 0$ holds within all $\mathcal{E}_n$.  Likewise, ensuring that $V$ is an ES-CLF can be achieved by ensuring that an upper-bound to (\ref{eq:stabilityy}) is non-positive for all $x \in \mathcal{C}(c)$. Further using GP estimates, one can show that  $\dot{\mathcal{V}}(x, \pi(x))\in [l_{\dot{\mathcal{V}}}^n(x, \pi(x)),\;u_{\dot{\mathcal{V}}}^n(x, \pi(x))]$ where
 
 \begin{align}
  & l_{\dot{\mathcal{V}}}^n(x, \pi(x))=\mu_{n,\dot{V}}+K_VV(x)-\sqrt{\beta_d^n} \sigma_{\dot{V},n}\\
  & u_{\dot{\mathcal{V}}}^n(x, \pi(x))=\mu_{n,\dot{V}}+K_VV(x)+\sqrt{\beta_d^n}(x) \sigma_{\dot{V},n}
 \end{align}
 and
\begin{align}
    &\mu_{n,\dot{V}}=\frac{\partial V(x)}{\partial x}^\top (f(x)+\mu_{d}^n(x)+g(x)u(x))\\
   & \sigma_{n, \dot{V}}=\bigg|\frac{\partial V(x)}{\partial x}^\top\bigg| \sigma_d^{n}(x).
\end{align}

Overall, we are thus led to determining a control law $\pi$ such that (1) $(\mathcal{E}_n\cap \mathcal{C}(c)) \cap \mathcal{X}$ is as large as possible and (2) $l_{\dot{\mathscr{H}}}^n(x, \pi(x))\geq 0$ and  $u_{\dot{\mathscr{V}}}^n(x, \pi(x))\leq 0$ holds within that region.
 %



However, it is not possible to check these latter two conditions for all $x\in \mathcal{C}(c)$. To circumvent this challenge we exploit the Lipschitz continuity properties of $\dot{\mathcal{V}}$ and $\dot{\mathscr{H}}$ to evaluate the sign of these certificates in a continuous domain  by evaluating $u^n_{\dot{\mathcal{V}}}\leq 0$ and $l_{\dot{\mathscr{H}}}^n\geq 0$ at a finite number of points.

With our goal and approach now hopefully made clear, we devote the next section to the construction of these approximate sets using Gaussian Process theory, while the final one considers discretization of the obtained conditions, under some Lipschitz-continuity assumption.

\section{Gaussian Process} \label{sec:Gaussian}

\subsection{Basic results}

In order to be able to learn the unknown term $d(x)$ using reproducing kernels we need some standard assumptions. First, we assume that $d(x)$ has low complexity as measured under the norm of a \textit{reproducing kernel Hilbert space} (RKHS). An RKHS $\mathcal{H}_k(\mathcal{X})$, which includes functions of the form $d(x)=\Sigma_i \alpha_i k(x,x_i)$ (where $\alpha_i\in \mathbb{R}$ and representer $x_i\in \mathcal{X}$), is a complete subspace of square-integrable function space, $\mathcal{L}_2$. $k(.,.)$ is symmetric, positive definite kernel function and $\langle d,k(x,.) \rangle=d(x)$ $\forall d(x)\in \mathcal{H}_k (\mathcal{X})$. Note that $\|d\|_k^2=\langle d, d\rangle_k$ is a measure for smoothness of $d(x)$ with respect to the kernel $k(.,.)$ as 

\begin{align*}
    |d(x)-d(y)|=&|\langle d, k(x,.)-k(y,.) \rangle|\\
    \leq &\|d\|_k\|k(x,:)-k(y,.)\|_k.
\end{align*}


We further make the following standard assumptions. 

\begin{assumption} \label{assum:measuablity}
 $\hat{d}(x)=\dot{x}-f(x)-g(x)u+\omega$ can be measured at all time, and the measurement noise $\omega$ is
 drawn independently from normal distribution $\mathcal{N}(0,\sigma^2)$ and is conditional $R-$sub-Gaussian.
 
\end{assumption}
This assumption is not restrictive as we only need to have access to $x$, $u$, and $\dot{x}$. The latter one is obtained by discrete-time approximation when the state $x$ is observable. Furthermore, conditional $R-$sub-Gaussian assumption for the noise $\omega_{t}$, defined by
\begin{align*}
   \mathbb{E}[e^{\lambda\omega_t}|\;x_{0:t},\omega_{0:t-1}]\leq e^{\frac{\lambda^ 2 R^2}{2}} \forall t\geq 0,\; \forall \lambda \in \mathbb{R},
\end{align*}
is a standard assumption in bandit literature (see \cite{abbasi2011improved, agrawal2013thompson, vakili2021optimal}) and controls community (see \cite{abbasi2011regret, cohen2019learning, lale2022reinforcement}).

\begin{assumption}
\label{assu2}
The kernel function  $k(x, x^{\prime})$ is continuously differentiable in both variables and bounded in $x\in \mathcal{X}$. In addition, there exists a RKHS upper bound for the function $d$, i.e., some $B_d>0$ such that $\|d\|_k\leq B_d$ .
\end{assumption}

For a noisy sample  $y_n=[\hat{d}(x_1),...,\hat{d}(x_n)]^\top$ at point $A_n=\{x_1,...,x_n\}$, and $U_n=\{u_1,...,u_n\}$, the posterior when $\hat{d}(x)=\dot{x}-f(x)-g(x)u+\omega$ with $\omega 	\sim \mathcal{N}(0,\sigma^2)$ is a GP distribution with $\mu^n_d(x)$, covariance $k^n_d(x,x^{\prime})$, and variance ${\sigma^n_d}(x)$:

\begin{equation}
\begin{array}{ll}
\mu^n_d(x)=k^n_d(x)^\top (K^n_d+\sigma ^2 I)^{-1}y_n \\
k^n_d(x,x^{\prime})=k(x,x^{\prime})-k^n_d(x)^\top (K^n_d+\sigma^2 I)^{-1}k^n_d(x^{\prime})
\\
{\sigma^n_d}(x)=\sqrt{k^n_d(x,x)}
\end{array} \label{GaussianD}
\end{equation}
where $k^n_d(x)=[k(x_1,x),...,k(x_n,x)]^\top$ and $K^n_d=[k(x,x^{\prime})]_{x,x^{\prime}}$ with $x,x^{\prime}\in A_n$.

As announced in the previous section, under Assumption \ref{assum:measuablity} and \ref{assu2}, it is possible to build a confidence set from these measurements, in which $d(x)$ is guaranteed to lie with high probability. This is the content of the following theorem.

\begin{theorem} 
Suppose that $\|d\|_k\leq B_d$ and let $\omega_t$ be $R$-gaussian noise and let $d:\mathcal{X}\rightarrow \mathbb{R}$ where $\mathcal{X} \subset \mathbb{R}$, then for all $n\geq 1$ and $x\in \mathcal{X}$ with probability at least $1-\delta$
\begin{align}
  \nonumber & |d(x)-\mu_d^{n-1}(x)|\leq\\ &(B_d+\alpha R\sqrt{2(\gamma_{n-1}+1+\ln (1/\delta))} \sigma^{n-1}_d(x)
  \label{BBB}
\end{align}
where $\gamma_{n-1}$ stands for the maximum information gain after $n-1$ rounds and $\alpha=1/\sigma ^2$.
\end{theorem}
\begin{proof}
The proof follows similar steps of \cite{chowdhury2017kernelized} with appropriate modifications. Indeed, while in \cite{chowdhury2017kernelized} the measurement noise is considered to be drawn independently from a normal distribution of the form $\mathcal{N}(0,\nu^2\lambda)$ with free parameters $\nu$ and $\lambda$, we work with assuming a noise $\mathcal{N}(0,\sigma^2)$, which is more  customary in the controls community (\cite{berkenkamp2016safe}, \cite{berkenkamp2017safe}, \cite{cheng2019end}). The necessary modifications are explicited in the proof of Theorem \ref{thm:confidenceSetDer}.

\end{proof}

The maximum information gain $\gamma_n$ appearing in (\ref{BBB}) is an information theoretic measure which is defined as follows:

\begin{align}
   \gamma_n=\max_{A\subset \mathcal {X}, |A|=n} I(y_A;d_A) \label{eq:maxInfGain}
\end{align}
where the mutual information gain 
\begin{align*}
   I(y_A;d_A)=H(y_A)-H(y_A|\; d_A),
\end{align*}
is the difference between the marginal entropy of measurements $y_A$, $H(y_A)$ and the conditional entropy $H(y_A|\; d_A)$ of observations given the valuation of $d(x)$ on the set $A$. The maximum information gain $\gamma_n$ is a measure representing the maximum reduction about uncertainty about function $d(x)$ after $n$ measurements $y_n$'s. This quantity is defined in the following lemma (a proof of which can be found, e.g., in \cite{KevinOnline} and \cite{chowdhury2017kernelized}).

\begin{lemma}
Let the measurements at points $A_n=\{x_1,...,x_n\}$  be $y_n=[\hat{d}(x_1),...,\hat{d}(x_n)]^\top$ and let $d_n=[d(x_1),...,d(x_n)]^\top$ denote the function values, then the information gain in these points is defined as follows:
\begin{align}
    I(y_n;d_n)=\frac{1}{2}\sum_{s=1}^{n} \ln (1+\frac{{\sigma_d^n}^2}{\sigma^2})
\end{align}
where ${\sigma_d^n}^2$ is given by (\ref{GaussianD}).
\end{lemma}

\subsection{Estimating derivatives and products from function measurements}

As mentioned above, we need to build a confidence set around $\frac{\partial d(x)}{\partial x}$ and $d(x)\frac{\partial d(x)}{\partial x}$ by just relying on the measurements of $d(x)$. To this end, we build on existing results in the bandit literature for GP \cite{srinivas2009gaussian, chowdhury2017kernelized} and provide the following sequence of results.

The first lemma gives the mean and covariance value of $\frac{\partial d(x)}{\partial x}$ relying merely on the data $y_n$.

\begin{lemma}
Suppose $d(x)$ has the posterior gaussian distribution (\ref{GaussianD}), then its first derivative with respect to $x$ has a gaussian distribution with mean $\mu_{\partial d}^{n-1}(x)$ and variance ${\sigma^n_{\partial d}}(x)$, defined as follows:

\begin{equation}
\begin{array}{ll}
\mu_{\partial d}^{n-1}(x)=(\frac{\partial k_d^n(x) }{\partial x}) ^\top (K^n_d+\sigma ^2 I)^{-1}y_n \\
{\sigma^n_{\partial d}}(x)=\sqrt{k^n_{\partial d}(x,x)}
\end{array} \label{eq:Moshtagh}
\end{equation}
where 
\begin{equation*}
 k^n_{\partial d}(x,x^{\prime})=\frac{\partial^2 k(x,x^{\prime})}{\partial x \partial x^{\prime}}-\frac{\partial k^{\top}_d(x)}{\partial x} (K+\sigma I)^{-1} \frac{\partial k_d(x)}{\partial x}.  \end{equation*}
\end{lemma}
\begin{proof}
For the derivation, see \cite{mchutchon2013differentiating}.
\end{proof}
Note that the superscripts $n$ or $n-1$ denote the number of samples or rounds.

Before providing a high probability confidence set for the estimate of $\partial d(x)/\partial x$, we need the following preliminaries. Recalling the reproducing property of kernel $k$, i.e., $\langle d,k(x,.) \rangle=d(x)$ $\forall d\in \mathcal{H}_k$ and $x\in \mathcal{X}$, we are interested to know the condition under which the reproducing property holds for the first derivative of $d(x)$, i.e., $\langle d,\frac{\partial k(x,.)}{\partial x} \rangle=\frac{\partial d(x)}{\partial x}$. First, we need the following definition

\begin{Definition}
(Mercer Kernel) Let $\mathcal{X}$ be a separable metric space. A kernel $k:\mathcal{X}\times \mathcal{X}\rightarrow \mathbb{R}$ is called a Mercer Kernel, if it is a continuous, symmetric and positive semi-definite function.
\end{Definition}

The following lemma is a specific case of Theorem 1 from \cite{zhou2008derivative} that gives required condition for reproducing property of $\partial d(x)/\partial x$.

\begin{lemma} \label{lem:DiifReprocuc}
Let $k:\mathcal{X}\times \mathcal{X}\rightarrow \mathbb{R}$ be a Mercer kernel such that, in addition,  $k$ is $C^{2}$ in both variables, then for any $x\in \mathcal{X}$, $\partial k(x,.)/\partial x\in \mathcal{H}_k$ and first order derivative reproducing property holds true, i.e.,
\begin{align}
    \frac{\partial d(x)}{\partial x}=\langle d,\frac{\partial k(x,.)}{\partial x} \rangle\quad \forall x\in \mathcal{X},\quad d\in \mathcal{H}_k.
\end{align}
\end{lemma}

We use this result to prove following theorem that gives high probability interval for the estimation of $\partial d(x)/\partial x$.
\begin{theorem} \label{thm:confidenceSetDer}
Suppose that $k\in C^2(\mathcal{X}\times \mathcal{X})$,  and let $\omega$ be zero mean and bounded by $\sigma$, then for all $n\geq 1$ and $x\in \mathcal{X}$ with probability at least $1-\delta$
\begin{align}
   \nonumber &|\frac{\partial d(x)}{\partial x}-\mu_{\partial d}^{n-1}(x)|\leq\\
    &(B_{d}+\alpha R\sqrt{2(\gamma_{n-1}+1+\ln (1/\delta))} \sigma_{\partial N}^{n-1}(x)
\end{align}
where $\alpha=1/\sigma ^2$.
\end{theorem}

\begin{proof}
Proof can be found in Appendix of \cite{chekan2022Safetyaware}.
\end{proof}






To estimate the $d(x)\frac{\partial d(x)}{\partial x}$, we first estimate $d^2(x)/2$ by GP (assuming $d^2(x)/2\in \mathcal{H}_k$), then obtain a high probability confidence interval for $d(x)\frac{ \partial d(x)}{\partial x}$ by following the same steps of Theorem 1, Lemma 2, and Theorem 3. For this, given Assumption \ref{assum:measuablity} we have access to the measurement of $\hat{d}^2(x)/2=(\dot{x}-f(x)-g(x)u)^2/2+\omega$, where $\omega 	\sim \mathcal{N}(0,\sigma^2)$ is a Gaussian measurement noise.  

For a noisy sample  $\bar{y}_n=[\hat{d}^2(x_1)/2,...,\hat{d}^2(x_n)/2]^\top$ at point $A_n=\{x_1,...,x_n\}$, and $U_n=\{u_1,...,u_n\}$. the posterior is a GP distribution with $\mu^n_{d^2/2}(x)$, covariance $k^n_d(x,x^{\prime})$, and variance ${\sigma^n_d}(x)$:

\begin{equation}
\begin{array}{ll}
\mu^n_{d^2/2}(x)=k^n_{d^2/2}(x)^\top (K^n_d+\sigma ^2 I)^{-1}\bar{y}_n \\
k^n_{d^2/2}(x,x^{\prime})=k(x,x^{\prime})-k^n_d(x)^\top (K^n_d+\sigma^2 I)^{-1}k^n_{d^2/2}(x^{\prime})
\\
{\sigma^n_d}(x)=\sqrt{k^n_{d^2/2}(x,x)}
\end{array} 
\end{equation}
where $k^n_d(x)=[k(x_1,x),...,k(x_n,x)]^\top$ and $K^n_d=[k(x,x^{\prime})]_{x,x^{\prime}}$ with $x,x^{\prime}\in A_n$. Applying same derivation as of Lemma 2, the estimate of $d(x)\frac{\partial d(x)}{\partial x}$ is a gaussian distribution with mean $\mu_{d\partial d}^{n-1}(x)$ and variance ${\sigma^n_{d\partial d}}(x)$ and defined as follows:

\begin{equation}
\begin{array}{ll}
\mu_{d\partial d}^{n-1}(x)=(\frac{\partial k_{d^2/2}^n(x) }{\partial x}) ^\top (K^n_{d^2/2}+\sigma ^2 I)^{-1}\bar{y}_n \\
{\sigma^n_{d\partial d}}(x)=\sqrt{k^n_{d\partial d}(x,x)}
\end{array} 
\end{equation}
where 
\begin{equation*}
 k^n_{\partial d}(x,x^{\prime})=\frac{\partial^2 k(x,x^{\prime})}{\partial x \partial x^{\prime}}-\frac{\partial k^{\top}_d(x)}{\partial x} (K+\sigma I)^{-1} \frac{\partial k_d(x)}{\partial x}.  
 \end{equation*}
 \begin{theorem}
Suppose that $\|d^2/2\|_k\leq B_{d^2/2}$ and that $k\in C^2(\mathcal{X}\times \mathcal{X})$, and let $\omega$ satisfy Assumption \ref{assum:measuablity}, then for all $n\geq 1$ and $x\in \mathcal{X}$ with probability at least $1-\delta$
\begin{align}
   \nonumber &|d(x)\frac{\partial d(x)}{\partial x}-\mu_{d\partial d}^{n-1}(x)|\leq\\
    &(B_{d^2/2}+\alpha R\sqrt{2(\bar{\gamma}_{n-1}+1+\ln (1/\delta))} \sigma_{\partial d}^{n-1}(x)
\end{align}
where $\alpha=1/\sigma^2$.
\end{theorem}
 \begin{proof}
Now by having mercer type of kernel $k_{d^2/2}\in C^2(\mathcal{X}\times \mathcal{X})$, and applying same steps of proof the claim holds true.
\end{proof}

\section{Exponential CLF-CBF Constraints in Discrete Domain}
\label{sec:DiscreteECLFandCBF}

The certificates obtained in the continuous domain are not useful in the control synthesis procedure, as checking them in the continuous domain is computationally burdensome. To circumvent this challenge, we exploit the Lipschitz continuity of the dynamics and the imposed constraint to introduce equivalent certificates in a discrete domain. These conditions enable the evaluation of the certificates in a finite number of points rather than the whole continuous domain throughout control design procedure. In other words, using the Lipschitz continuity properties stated by Assumption \ref{ass_lip}, we generalize our knowledge about safety into states we have not explored yet. 
\begin{assumption} \label{ass_lip}
The dynamics $f(.)$, $g(.)$, $h(.)$,  $\mathcal{F}(.)$, and $\partial\mathcal{F}(.)/\partial (.)$ are ${L}_f-$, ${L}_g-$, ${L}_h-$, ${L}_h-$, ${L}_{\mathcal{F}}-$, $L_{\partial \mathcal{F}}-$ Lipschitz continuous. Furthermore, to keep the closed loop Lipschitz continuous we restrict the control policy $\pi(.)$ to be in the set of all $L_{\pi}-$ Lipschitz functions.
\end{assumption}
Under these assumptions, the following lemmas, whose proofs have been provided in Appendix of \cite{chekan2022Safetyaware}, give the Lipschitz constants $L_{\dot{\mathscr{H}}}$ and $L_{\dot{\mathcal{V}}}$ of $\dot{\mathcal{V}}$ and $\dot{\mathcal{H}}$ statements.

\begin{lemma} \label{lem:lipstchitzH}
    The function $\dot{\mathscr{H}}(x,\pi(x))$ is Lipschitz continuous with constant $L_{\dot{\mathscr{H}}}$,
      \begin{align}
  \nonumber L_{\dot{\mathscr{H}}}=&(B_f+B_gB_{\pi}+B_d\|k\|_{\infty})(L_{\partial h}+L_dL_{\partial \mathcal{F}}+L_{\mathcal{F}}L_{\partial d})\\
   \nonumber&+(L_h+L_{\mathcal{F}}B_d\|\frac{\partial k}{\partial x}\|_{\infty})(L_f+L_d+B_gL_{\pi})\\
   &+|K_H|(L_h+L_{\mathcal{F}}L_d))|x-x^\prime|\label{eq:lipH}
    \end{align}
    where $L_V=\|\frac{\partial V}{\partial x}\|_{\infty}$ and $L_{\partial V}=\|\frac{\partial^2 V}{\partial x^2}\|_{\infty}$. \label{eq:lipH}
\end{lemma}

\begin{lemma} \label{lem:Vliptchitz}
    The function $\dot{\mathcal{V}}(x,\pi (x))$ is Lipschitz continuous with constant $L_{\dot{\mathcal{V}}}$,
    \begin{align}
  \nonumber L_{\dot{\mathcal{V}}}=&(B_f+B_gB_{\pi}+B_d\|k\|_{\infty}  )L_{\partial V}\\
  &+(L_f+L_d+B_gL_{\pi}+K_V)L_V \label{eq:lipV}
    \end{align}
    
    where $L_V=\|\frac{\partial V}{\partial x}\|_{\infty}$ and $L_{\partial V}=\|\frac{\partial^2 V}{\partial x^2}\|_{\infty}$. 
\end{lemma}

In the statements (\ref{eq:lipH}) and (\ref{eq:lipV}), we still need to define $L_d$, $L_{\partial d}$ using kernel properties. The following lemma gives these constants. 

\begin{lemma} \label{lem:Liptch const d, dd}
    The functions $d(x)$ and $\frac{\partial d(x)}{\partial x}$ are Lipschitz continuous with constants $L_d$ and $L_{\partial d}$ 
    \begin{align}
      L_d =B_d\|\frac{\partial k}{\partial x}\|_{\infty}, \;\;L_{\partial d}=B_d\|\frac{\partial ^2k}{\partial x^2}\|_{\infty}
    \end{align}
\end{lemma}

With functions  $\dot{\mathcal{V}}$ and $\dot{\mathcal{H}}$'s Lipschitz constants in hand, we now proceed to derive the equivalent discrete domain certificates for control synthesis purpose. 

We let $\mathcal{X}_{\tau}$ denote discretization of continuous space $\mathcal{X}$ such that $|x-[x]_{\tau}|\leq \tau/2$ where $x\in \mathcal{X}$ and $[x]_{\tau}\in \mathcal{X}_{\tau}$. Then it is trivial to write:
\begin{align}
   \nonumber & |\dot{\mathcal{V}}(x, \pi(x))-\dot{\mathcal{V}}([x]_{\tau}, \pi([x]_{\tau}))|\leq L_{\dot{\mathcal{V}}}\tau\\
   &  |\dot{\mathcal{H}}(x, \pi(x))-\dot{\mathcal{H}}([x]_{\tau}, \pi([x]_{\tau}))|\leq L_{\dot{\mathcal{H}}}\tau. \label{eq:lipineq}
\end{align}


Given the discretization, and using the Lipschitz continuity of the dynamics and safety constraint, the following theorem summarizes the generalization of decrease condition on $\mathcal{V}$ and increase condition on $\mathcal{H}$ from discrete space $\mathcal{X}_{\tau}$ to continuous space $\mathcal{X}$.

\begin{theorem} \label{tm.Disctetization}
    Let $\mathcal{X}_{\tau}$ be a discretization of $\mathcal{X}$ and let
    \begin{align}
    \bar{\mathcal{E}}_n(e):= \bigg\{x\in \mathbb{R}\; |\;h(x)+ \min_{y \in I^n_{d}(x)} \mathcal{F}(y) \geq e \bigg\}.
    \end{align}   
 Assume that 
      \begin{align}
    &\nonumber u^n_{\dot{\mathcal{V}}}(x, \pi(x))\leq -L_{\dot{\mathcal{V}}}\tau\\
    &l^n_{\dot{\mathcal{H}}}(x, \pi(x))\geq L_{\dot{\mathcal{H}}}\tau.
    \end{align}   
    hold for all $x\in \big(\bar{\mathcal{E}}_n(e)\cap \mathcal{C}(c)\big)\cap \mathcal{X}_{\tau}$ with $c>0$, $e\geq 0$ and $n\in \mathbb{N}$, and $u=\pi(x)$, then $\dot{\mathcal{V}}(x, \pi(x))\leq 0$ and $\dot{\mathcal{H}}(x, \pi(x))\geq 0$ hold for $x\in\mathcal{E}\cap \mathcal{C}(c)$ with probability at least $1-\delta$. In other words $\mathcal{E}\cap \mathcal{C}(c)$ is the intersection of safe set and region of attraction for system (\ref{eq:model}) under policy $\pi$ with high probability.
\end{theorem}

\vspace{0.3cm}

Now, we are in the position to briefly illustrate one iteration of control design procedure, consisting of policy optimization and exploration phases. With $n$ number of measurements, we define the set of all state-action pairs $\mathcal{D}_n$ that satisfy the safety certificates as follows:
\begin{align}
    \mathcal{D}_n=\{(x,u)\in \mathcal{X}_{\tau}\times \mathcal{U}|\;&u^n_{\dot{\mathcal{V}}}(x, \pi(x))\leq -L_{\dot{\mathcal{V}}}\tau,\\
    \nonumber &l^n_{\dot{\mathcal{H}}}(x, \pi(x))\geq L_{\dot{\mathcal{H}}}\tau \}
\end{align}
using which we compute a control policy such that the size of intersected safe set and region of attraction,  $\mathcal{C}(c_n)\cap \bar{\mathcal{E}}_n(e_n)$ to be maximum. We refer to this procedure as policy optimization phase which is mathematically formulated by the following optimization problem:
\begin{align}\label{eq:Thm4}
   \nonumber \pi_n, c_n,e_n&= \operatorname*{argmax}_{c,e\in \mathbb{R}_{\geq 0}, \pi \in \prod_L}|\bar{\mathcal{E}}_{n}(e)\cap \mathcal{C}(c)|\\
   \nonumber &\textrm{s.t.}\;\; \forall x\in  \big(\bar{\mathcal{E}}_{n}(e_n)\cap \mathcal{C}(c_n)\big)\cap \mathcal{X}_{\tau},\\
 & (x,\pi(x))\in \mathcal{D}_{n}
\end{align}
where $\prod_L$ is the class of all possible policies and $|\bullet|$ stands for the cardinality of set $\bullet$. It is worthy to note that the problem (\ref{eq:Thm4}), thanks to Theorem  \ref{tm.Disctetization}, is computationally tractable and its solution $\pi_n$ guarantees the main goals: stabilization and forward invariance of the safety set.

By having in hand the policy $\pi_n$ and $\mathcal{S}_n:=\bar{\mathcal{E}}_{n}(e_n)\cap \mathcal{C}(c_n)$ through solving (\ref{eq:Thm4}), the next phase of control design is exploration, in which we aim to shrink the uncertainity of confidence intervals in $\mathcal{S}_n$ in order to expand the intersected region further, i.e., $|\mathcal{S}_{n+1}|>|\mathcal{S}_n|$. For this, we need to drive the system to the state $x_{n+1}\in \mathcal{S}_n$ in which we are less certain about the unknown term of the dynamics or equivalently about $\dot{\mathcal{V}}(x,u)$ and $\dot{\mathcal{H}}(x,u)$. 
We define the uncertainities of these estimates as follows
\begin{align*}
    &\Delta \dot{\mathcal{V}}(x,\pi_n(x))=u^{n}_{\dot{\mathcal{V}}}(x, \pi_n(x))-l^{n}_{\dot{\mathcal{V}}}(x, \pi_n(x))\\
    & \Delta \dot{\mathcal{H}}(x,\pi_n(x))=u^{n}_{\dot{\mathcal{H}}}(x, \pi_n(x))-l^{n}_{\dot{\mathcal{H}}}(x, \pi_n(x))
\end{align*}
which are size of corresponding confidence intervals.
Therefore, the state of interest $x_{n+1}$ to be visited is obtained by solving
\begin{align}
    x_{n+1}=\operatorname*{argmax}_{x\in \mathcal{S}_n}\big(\Delta \dot{\mathcal{V}}(x,\pi_n(x))+\Delta\dot{\mathcal{H}}(x,\pi_n(x))\big) \label{eq:BestNextState}
\end{align}
and in exploration phase by applying the backup policy $\pi_n$ we drive the system there. By letting $\mathcal{R}_{\pi_n}$ denote the true but unknown intersected safe region and region of attraction under the policy $\pi_n$, by Theorem \ref{tm.Disctetization}, it is straight forward to show that $\mathcal{S}_n\subseteq \mathcal{R}_{\pi_n}$ for any $n\geq 1$. For the sake of brevity we skip providing the detailed analysis here.

\vspace{-0.25cm}
\section{Summary and Conclusion} \label{Sec:Conclus}
This work addressed the problem of safely controlling a nonlinear, partially unknown system while guaranteeing prescribed safety constraints. We assumed that the uncertainty of dynamics could explicitly affect this prescribed set. Having access to noisy measurement of the uncertain part of dynamics and using GPs and reproducing kernel properties, we built confidence intervals for the unknown dynamics term, its first-order derivative, and their products. Proving that these confidence intervals estimate those terms with high probability, we used them to derive Es-CLF and ECBF certificates that, with high probability, guarantee stability and forward invariance of the prescribed safety set. Since these certificates are in a continuous domain, checking which in the whole domain may not be computationally efficient, we exploited the Lipschitz continuity properties of the system and, as such, derived equivalent conditions in the discrete domain. Finally, using the obtained certificates, we proposed a strategy that designs a control that fulfills our goals throughout learning the uncertainties while enlarging the intersection of ROA and the imposed-safety set. Future work will consider extending our analysis to high-dimensional system, in addition to explicitly addressing the simplifications introduced for brevity in the current paper.

\bibliographystyle{IEEEtran}
\bibliography{ref}

\newpage
\section*{APPENDIX} \label{append}
\textbf{Proof of Theorem \ref{thm:confidenceSetDer}}

\begin{proof}
Noting that $y_{1:n}=d_{1:n}+\omega_{1:n}$ and Using (\ref{eq:Moshtagh}), we can write:
\begin{align}
    &|\frac{\partial d(x)}{\partial x}-\mu_{\partial d}^{n-1}(x)|\\
   \nonumber  &=|\frac{\partial d(x)}{\partial x}-(\frac{\partial k_{n-1}(x) }{\partial x}) ^\top (K^n_d+\sigma ^2 I)^{-1}y_n|\\
   \nonumber  &\leq \underbrace{|\frac{\partial d(x)}{\partial x}-(\frac{\partial k_{n-1}(x) }{\partial x}) ^\top (K^n_d+\sigma ^2 I)^{-1}d_{1:n}|}_\text{$\Gamma_1$}\\
   \nonumber  &+\underbrace{|(\frac{\partial k_{n-1} }{\partial x}) ^\top (K^n_d+\sigma ^2 I)^{-1}\omega_{1:n}|}_\text{$\Gamma_2$}
\end{align}
By RKHS property we have $d(x)=\langle d,\ k(x,.) \rangle_k$ and letting $\phi (x)=k(x,)$ be feature map, then one can write $d(x)=d^\top \phi(x)$. Similarly, $k(x,x^{\prime})=\langle k(x,.),\ k(x^{\prime},.)\rangle_k=\phi(x)^\top \phi(x^{\prime})$ for all $x,x^{\prime}\in \mathcal{X}$. We also have $K_d^{n-1}=\Phi_{n-1}^\top\Phi_{n-1}$ where $\Phi_{n-1}=[\phi(x_1)^\top,...,\phi_(x_{n-1})^\top]^\top$. Likewise, $k_n(x)=\Phi_n\phi(x)$ for $x\in \mathcal{X}$ and $d_{1:n}=\Phi_nd$.

Now we first upper-bound the term $\Gamma_1$. Assuming that $\frac{\partial d}{\partial x}\in \mathcal{H}_k$, we have 
\begin{align}
   {\Gamma_1}&=|\frac{\partial d(x)}{\partial x}-\frac{\partial \phi}{\partial x}^\top\Phi_n^\top(\sigma^2 I+\Phi_n\Phi_n^\top )^{-1}\Phi_n d|\\
  \nonumber &=|\frac{\partial d(x)}{\partial x}-\frac{\partial \phi}{\partial x}^\top(\sigma^2 I+\Phi_n^\top\Phi_n )^{-1}\Phi_n^\top \Phi_n d|\\
  \nonumber &= |\frac{\partial d(x)}{\partial x}-\frac{\partial \phi}{\partial x}^\top d+\frac{\partial \phi}{\partial x}^\top(\sigma^2 I+\Phi_n^\top\Phi_n )^{-1}\sigma^2 d|\\
  \nonumber &\leq \underbrace{|\frac{\partial d(x)}{\partial x}-\frac{\partial \phi}{\partial x}^\top d|}_\text{$\Gamma_{11}$}+\underbrace{|\frac{\partial \phi}{\partial x}^\top(\sigma^2 I+\Phi_n^\top\Phi_n )^{-1}\sigma^2 d|}_\text{$\Gamma_{12}$}
  \end{align}
where is the second equality we applied 
\begin{align}
   \Phi_n^\top(\sigma^2 I+\Phi_n\Phi_n^\top )^{-1}=(\sigma^2 I+\Phi_n^\top\Phi_n )^{-1}\Phi_n^\top \label{eq:swap}
\end{align}
resulting from the fact that $(\sigma^2 I+\Phi_n\Phi_n^\top )$ is strictly positive definite. 

By mercer type of kernel assumption and applying Lemma \ref{lem:DiifReprocuc}, we have $\Gamma_{11}=0$.  



To bound the term $\Gamma_{12}$, we need following result first. One can write:

\begin{align*}
 (\sigma^2 I+\Phi_n^\top\Phi_n) \frac{\partial \phi(x)}{\partial x}= \Phi_n^\top \frac{\partial k_n}{\partial x} +\sigma^2 \frac{\partial \phi(x)}{\partial x}
\end{align*}
Multiplying both sides by $(\sigma^2 I+\Phi_n^\top\Phi_n)^{-1}$ and then applying (\ref{eq:swap}) results in

\begin{align}
    \frac{\partial \phi(x)}{\partial x}=&\Phi_n^\top(\sigma^2 I+\Phi_n\Phi_n^\top )^{-1} \frac{\partial k_n(x)}{\partial x}+\\
    \nonumber&+\sigma^2 (\sigma^2 I+\Phi_n^\top\Phi_n)^{-1}\frac{\partial \phi(x)}{\partial x}
\end{align}
which holds for any $x, x^{\prime}\in \mathcal{X}$. Then noting that
\begin{align}
  \frac{\partial ^2k(x,x^\prime)}{\partial x \partial x^{\prime}}= \frac{\partial\phi(x)}{\partial x}^\top \frac{\partial\phi(x^{\prime})}{\partial x^{\prime}} 
\end{align}
it yields 

\begin{align}
    \nonumber \frac{\partial ^2k(x,x^\prime)}{\partial x \partial x^{\prime}}=&\frac{\partial k_n(x)}{\partial x}^\top(\sigma^2 I+\Phi_n\Phi_n^\top )^{-1}\frac{\partial k_n(x^{\prime})}{\partial x^\prime}\\
   & \sigma^2 \frac{\phi(x)}{\partial x}^\top(\sigma^2 I+\Phi_n^\top\Phi_n )^{-1} \frac{\phi(x^{\prime})}{\partial x^{\prime}}. \label{eq:important}
\end{align}
Now we upper-bound the term $\Gamma_{12}$:
\begin{align}
    \nonumber \Gamma_{12}\leq &\|d\|_k \|(\sigma^2 I+\Phi_n^\top\Phi_n )^{-1}\sigma^2 \frac{\partial \phi(x)}{\partial x}\|_k\\
   \nonumber \leq & B_k\bigg(\sigma^2\frac{\partial \phi(x)}{\partial x}^\top(\sigma^2 I+\Phi_n^\top\Phi_n )^{-1}\\
\nonumber& \sigma^2 (\sigma^2 I+\Phi_n^\top\Phi_n )^{-1} \frac{\partial \phi(x)}{\partial x}\bigg)^{1/2}\\
\nonumber \leq& B_k\bigg(\sigma^2\frac{\partial \phi(x)}{\partial x}^\top(\sigma^2 I+\Phi_n^\top\Phi_n )^{-1}\\
\nonumber& (\sigma^2 I+\Phi_n^\top\Phi_n ) (\sigma^2 I+\Phi_n^\top\Phi_n )^{-1} \frac{\partial \phi(x)}{\partial x}\bigg)^{1/2}\\
=& B_d\sigma^n_{\partial d}
\end{align}
where in the last equality we applied (\ref{eq:important}) and the definition of $\sigma^n_{\partial d}$.

The upper bound for $\Gamma_2$ is obtained as follows:

\begin{align}
  \nonumber &\Gamma_2= |\frac{\partial \phi(x)}{\partial x}^\top \Phi_n^\top (\Phi_n\Phi_n^\top+\sigma^2 I)^{-1}\omega_{1:n}|\\
  \nonumber &=|\frac{\partial \phi(x)}{\partial x}^\top  (\Phi_n^\top\Phi_n+\sigma^2 I)^{-1}\Phi_n\omega_{1:n}|\\
  \nonumber &\leq \|(\Phi_n^\top\Phi_n+\sigma ^2 I)^{-\frac{1}{2}}\frac{\partial \phi(x)}{\partial x}\|_k\|(\Phi_n^\top\Phi_n+\sigma^2 I)^{-\frac{1}{2}}\Phi_n^\top \epsilon_{1:n}\|_k\\
  \nonumber & =\sqrt{\frac{\partial \phi(x)}{\partial x}^\top (\Phi_n^\top \Phi_n+\sigma^2 I)^{-1}\frac{\partial \phi(x)}{\partial x}}\times \\
   \nonumber & \sqrt{(\phi_n^\top \omega_{1:n})^\top (\phi_n \Phi_n^\top+\sigma^2 I)^{-1}\phi_n^\top \omega_{1:n}}\\
   \nonumber &= {\sigma}^{-1} \sigma_{\partial d}^n(x)\sqrt{\omega_{1:n}^\top \Phi_n\Phi_n^\top (\phi_n\Phi_n^\top+\sigma^2 I)^{-1} \omega_{1:n}}\\
   \nonumber &= {\sigma}^{-1} \sigma_{\partial d}^n(x) \sqrt{\omega_{1:n}^\top K_n(K_n+\sigma^2 I)^{-1}\omega_{1:n}}
   \end{align}
   For an $\alpha$ such that $\alpha \sigma^2=\lambda<1$ it yields 
   \begin{align*}
    & {\sigma}^{-1} \sigma_{\partial d}^n(x) \sqrt{\omega_{1:n}^\top K_n(K_n+\sigma^2 I)^{-1}\omega_{1:n}}\leq \\
    &  \alpha{\lambda}^{-1} \sigma_{\partial d}^n(x) \sqrt{\omega_{1:n}^\top K_n(K_n+\lambda I)^{-1}\omega_{1:n}}\leq \\
    & \alpha\sigma_{\partial d}^n(x) \sqrt{\omega_{1:n}^\top K_n(K_n+\lambda I)^{-1}\omega_{1:n}}
   \end{align*}
 Letting $\lambda=1+\eta$ by following similar steps of \cite{chowdhury2017kernelized} we can write:
 \begin{align*}
     &\omega_{1:n}^\top K_n(K_n+(1+\eta) I)^{-1}\omega_{1:n}\leq \\
     &\omega_{1:n}^\top ((K_n+\eta I)^{-1}+ I)^{-1}\omega_{1:n}
 \end{align*}
 By applying the self-normalized concentration inequality (see Theorem 1 in \cite{chowdhury2017kernelized}), it yields
 \begin{align*}
    &\sqrt{ \omega_{1:n}^\top ((K_n+\eta I)^{-1}+ I)^{-1}\omega_{1:n}}\leq \\
    & R\sqrt{2\ln \frac{\sqrt{\det \big((1+\eta)I +K_n\big)}}{\delta}}
 \end{align*}
 Using the fact that
 \begin{align*}
      \det \big(\lambda I+K_n\big)=\det\big(I+\lambda^ {-1}K_n\big)\det\big(\lambda I\big)
 \end{align*}
it yields that
 \begin{align*}
      &\ln \sqrt{\det \big((1+\eta)I +K_n\big)}=\\
      & \frac{1}{2}\ln \bigg(\det\big(I+(1+\eta)^{-1}K_n\big)\bigg)+\frac{1}{2}n \ln (1+\eta)\leq\\
      & \frac{1}{2}\ln \bigg(\det\big(I+\sigma^{-2}K_n\big)\bigg)+\frac{1}{2}n \eta\leq \gamma_n+\frac{1}{2}\eta n
 \end{align*}
 where in the second inequality we applied the fact that $\lambda>\sigma^2$ and in third inequality we used the definition of information gain by (\ref{eq:maxInfGain}).
 By choosing $\eta=2/T$ where $T$ without loss of generality where $T$ is the horizon and $T>n$ it yields
\begin{align}
    \Gamma_2\leq \alpha \sigma_n(x)  R\sqrt{2(1+\gamma_n+\ln(1/\delta)}.
\end{align}
which completes the proof.
\end{proof}

\textbf{Proof of Lemma \ref{lem:Vliptchitz}}

\begin{proof} 
By having definition of $\dot{\mathcal{V}}$ by Lemma 4, one can write
\begin{align*}
 \bigg|&\frac{\partial V(x)}{\partial x}(f(x)+g(x)\pi(x)+d(x))+K_VV(x)\\
 &-\frac{\partial V(x^\prime)}{\partial x^\prime}(f(x^\prime)+g(x^\prime)\pi(x^\prime)+d(x^\prime))-K_VV(x^\prime)\bigg|\leq \\
 &\bigg|f(x)+g(x)\pi(x)+d(x)\bigg|\;\;\bigg|\frac{\partial V(x)}{\partial x}-\frac{\partial V(x^\prime)}{\partial x^\prime}\bigg|\\
 &+\bigg|\frac{\partial V(x^\prime)}{\partial x^\prime}\bigg|\;\;\bigg|f(x)+g(x)\pi(x)+d(x)\\
 &-f(x^\prime)-g(x^\prime)\pi(x^\prime)-d(x^\prime)\bigg|+K_V\bigg|V(x)-V(x^\prime)\bigg|\leq \\
 &(B_f+B_gB_{\pi}+B_d\|k\|_{\infty})L_{\partial V}(x-x^\prime)+\\
 &L_V(L_f+L_d+B_gL_{\pi}+K_V)(x-x^\prime)
\end{align*}
where we applied
\begin{align*}
    d(x)&=\langle d, k(x,.)\rangle_k\leq \|d\|_k \|k(x,.)\|_k\\
    &=\|d\|_k \sqrt{k(x,x)}\leq \|k\|_{\infty}\|f\|_k
\end{align*}
\end{proof}
\textbf{Proof of Lemma \ref{lem:Liptch const d, dd}}

\begin{proof} 
for $d(x)$ one can write
\begin{align*}
  &|d(x)-d(x^\prime)|=|\langle d,k(x,.) \rangle-\langle d,k(x^\prime,.) \rangle| \leq \\
  & \|d\|_k\|k(x,.)-k(x^\prime,.)\|_k\leq\|d\|_k \sqrt{k(.,x-x^\prime)k(.,x-x^\prime)}\\
  &\leq\|d\|_k \|\frac{\partial k}{\partial x}\|_{\infty}|x-x^\prime|\leq B_d\|\frac{\partial k}{\partial x}\|_{\infty} |x-x^\prime|.
\end{align*}
As for $\frac{\partial d(x)}{\partial x}$, having the reproducing property for $\frac{\partial d(x)}{\partial x}$ by Lemma \ref{lem:DiifReprocuc} we can write
\begin{align*}
  &|\frac{\partial d(x)}{\partial x}-\frac{\partial d(x^\prime)}{\partial x^\prime}|=|\langle d,\frac{\partial k(x,.)}{\partial x} \rangle-\langle d,\frac{\partial k(x^\prime,.)}{\partial x^\prime } \rangle| \leq \\
  & \|d\|_k\sqrt{\langle \frac{\partial k(x,.)}{\partial x}-\frac{\partial k(x^\prime,.)}{\partial x^\prime}, \frac{\partial k(x,.)}{\partial x}-\frac{\partial k(x^\prime,.)}{\partial x^\prime}\rangle}\leq \\
  & B_d\sqrt{\|\frac{\partial ^2k}{\partial x^2}\|_{\infty}^2|x-x^\prime|^2}\leq B_d\|\frac{\partial ^2k}{\partial x^2}\|_{\infty} |x-x^\prime|
\end{align*}
\end{proof}

\textbf{Proof of Lemma \ref{lem:lipstchitzH}}

\begin{proof}
By having definition of $\dot{\mathscr{H}}(x,\pi(x))$ by (\ref{SafetyCond}), one can write
\begin{align*}
    \bigg|&\big(\frac{\partial h(x)}{\partial{x}}+\frac{\partial \mathcal{F}(d(x))}{\partial{d(x)}} \frac{\partial d(x)}{\partial{x}}\big)\bigg(f(x)+d(x)+g(x)\pi(x)\bigg)\\
    &+K_HH(x)-\\
    &\big(\frac{\partial h(x^\prime)}{\partial{x^\prime}}+\frac{\partial \mathcal{F}(d(x^\prime))}{\partial{d(x^\prime)}} \frac{\partial d(x^\prime)}{\partial{x^\prime}}\big)\times\\
    &\bigg(f(x^\prime)+d(x^\prime)+g(x^\prime)\pi(x^\prime)\bigg)-K_HH(x^\prime)\bigg| \\
    &\leq\big|f(x)+d(x)+g(x)\pi(x)\big|\bigg|\frac{\partial h(x)}{\partial x}-\frac{\partial h(x^\prime)}{\partial{x^\prime}}\\
    &+\frac{\partial \mathcal{F}(d(x))}{\partial{d(x)}} \frac{\partial d(x)}{\partial{x}}-\frac{\partial \mathcal{F}(d(x^\prime))}{\partial{d(x^\prime)}} \frac{\partial d(x)}{\partial{x}}\\
    &+\frac{\partial \mathcal{F}(d(x^\prime))}{\partial{d(x^\prime)}} \frac{\partial d(x)}{\partial{x}}-\frac{\partial \mathcal{F}(d(x^\prime))}{\partial{d(x^\prime)}} \frac{\partial d(x^\prime)}{\partial{x^\prime}}\bigg|+\\
    & \bigg|\frac{\partial h(x^\prime)}{\partial{x^\prime}}+\frac{\partial \mathcal{F}(d(x^\prime))}{\partial{d(x^\prime)}} \frac{\partial d(x^\prime)}{\partial{x^\prime}}\bigg| \big|f(x)+d(x)+g(x)\pi(x)-\\
    &f(x^\prime)-d(x^\prime)-g(x^\prime)\pi(x^\prime)\big|+\\
    &|K_H|\bigg|h(x)+\mathcal{F}(d(x))-h(x^\prime)-\mathcal{F}(d(x^\prime))\bigg| \leq \\
   & (B_f+B_gB_{\pi}+B_d\|k\|_{\infty})(L_{\partial h}+L_dL_{\partial \mathcal{F}}+L_{\mathcal{F}}L_{\partial d})|x-x^\prime|\\
   &+(L_h+L_{\mathcal{F}}B_d\|\frac{\partial k}{\partial x}\|_{\infty})(L_f+L_d+B_gL_{\pi})|x-x^\prime|+\\
   &|K_H|(L_h+L_{\mathcal{F}}L_d))|x-x^\prime|
    \end{align*}
\end{proof}

\textbf{Proof of Theorem \ref{tm.Disctetization}}

\begin{proof}

Note that the conditions:

\begin{align*}
u^n_{\dot{\mathcal{V}}}(x, \pi(x)) &\leq -L_{\dot{\mathcal{V}}}\tau, \
l^n_{\dot{\mathcal{H}}}(x, \pi(x)) &\geq L_{\dot{\mathcal{H}}}\tau.
\end{align*}

hold for all $x \in \left(\bar{\mathcal{E}}n(e) \cap \mathcal{C}(c)\right) \cap \mathcal{X}\tau$, where $c>0$, $e\geq 0$, $n\in \mathbb{N}$, and $u=\pi(x)$. These conditions are sufficient to ensure that $u^n_{\dot{\mathcal{V}}}(x, \pi(x)) \leq 0$ and $l^n_{\dot{\mathcal{H}}}(x, \pi(x)) \geq 0$. Consequently, they guarantee the forward invariance of $\mathcal{E} \cap \mathcal{C}(c)$ in the continuous domain $\mathcal{X}$.

Thus, the proof is complete.

\end{proof}

\end{document}